\documentclass[envcountsame]{llncs}

\usepackage[english]{babel}
\usepackage[utf8]{inputenc}
\usepackage{graphicx,subcaption,url,hyperref,multirow,xspace,color} 
\usepackage{amsmath,amssymb,clrscode,thmtools,thm-restate,microtype,verbatim}
\usepackage{algorithm}
\usepackage[noend]{algorithmic}
\usepackage[makeroom]{cancel}
\usepackage{csquotes}
\usepackage{siunitx}

\usepackage{pdflscape}

\newcommand{\C}{\ensuremath{\mathcal{C}}\xspace}

\renewcommand{\S}{\ensuremath{\mathcal{S}}\xspace}
\newcommand{\T}{\ensuremath{\mathcal{T}}\xspace}
\newcommand{\F}{\ensuremath{\mathcal{F}}\xspace}

\graphicspath{{figures/}}

\newcommand{\fd}{flow diagram\xspace}
\newcommand{\Fd}{Flow diagram\xspace}
\newcommand{\FD}{Flow Diagram\xspace}

\title{Compact \FD{}s for State Sequences}
\author{Kevin Buchin \and Maike Buchin \and Joachim Gudmundsson \and \\Michael
Horton \and Stef Sijben}
\institute{}

\begin{document}
\maketitle

\begin{abstract}
We introduce the concept of compactly representing a large number of state
sequences, e.g., sequences of activities, as a \fd{}.
We argue that the \fd{} representation
gives an intuitive summary that allows the user to detect patterns among large
sets of state sequences. Simplified, our aim is to generate a small \fd{}
that models the flow of states of all the state sequences given as input.
For a small number of state sequences we present efficient algorithms to compute
a minimal \fd{}.
For a large number of state sequences we show that it is unlikely that efficient
algorithms 
exist. More specifically, the
problem is $W[1]$-hard if the number of state sequences is taken as a
parameter. We thus introduce several heuristics for this problem.
We argue about the
usefulness of the \fd{} by applying the algorithms to two problems in sports
analysis. We evaluate the performance of our algorithms on a football data
set and generated data.
\end{abstract}

\section{Introduction}
Sensors are tracking the activity and movement of an increasing number of objects, generating large data sets in many application domains, such as sports analysis, traffic analysis and behavioural ecology.
This leads to the question of how large sets of sequences of activities can be represented compactly. We introduce the concept of representing the ``flow" of activities in a compact way and argue that this is helpful to detect patterns in large sets of state sequences. 

To describe the problem we start by giving a simple example. Consider three objects (people) and their sequences of states, or activities, during a day. The set of state sequences $\T=\{\tau_1, \tau_2, \tau_3\}$ are shown in Fig.~\ref{fig:ex}(a).
As input we are also given a set of criteria $\C=\{C_1, \ldots, C_k\}$, as listed in Fig.~\ref{fig:ex}(b). Each criterion is a Boolean function on a single subsequence of states, or a set of subsequences of states.
For example, in the given example the criterion $C_1=$“eating” is true for Person~1 at time intervals 7--8am and 7--9pm, but false for all other time intervals.
Thus, a criterion partitions a sequence of states into subsequences, called \emph{segments}.
In each segment the criterion is either true or false. A \emph{segmentation}
of $\T$ is a partition of each sequence in $\T$ into true segments, which is represented by the corresponding sequence of criteria.
If a criterion $C$ is true for a set of subsequences, we say they \emph{fulfil} $C$.
Possible segments of $\T$ according to the set $\C$ are shown in~Fig.~\ref{fig:ex}(c). The aim is to summarize segmentations of all sequences efficiently; that is, build a \fd{} $\F$, starting at a start state $s$ and ending at an end state $t$, with a small number of nodes such that for each sequence of states $\tau_i$, $1\leq i\leq m$, there exists a segmentation according to $\C$ which appears as an $s$--$t$ path in $\F$.
A possible \fd{} is shown in Fig.~\ref{fig:ex}(d). This \fd{} for $\T$ according to $\C$ can be validated  by going through a segmentation of each object while following a path in $\F$ from $s$ to $t$. For example, for Person~1 the $s$--$t$ path $s\rightarrow C_1 \rightarrow C_2 \rightarrow C_4 \rightarrow C_1 \rightarrow t$ is a valid segmentation.

\begin{figure}[tb]
\includegraphics[width=\textwidth]{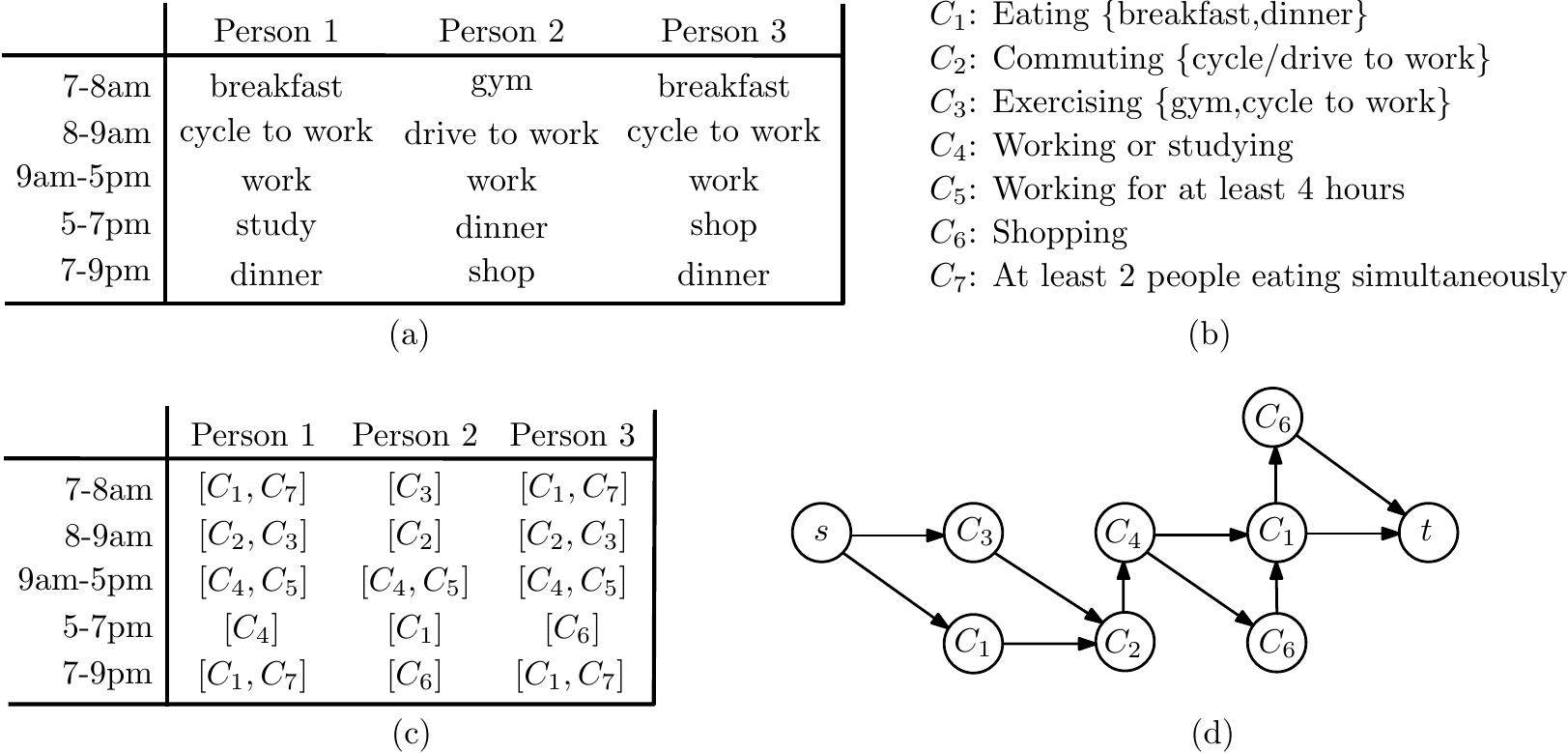}
\caption{The input is (a) a set $\T =\{\tau_1, \ldots, \tau_m\}$ of sequences of states and (b) a set of criteria $\C=\{C_1, \ldots, , C_k\}$. (c) The criteria partition the states into a segmentation. (d) A valid \fd{} for $\T$ according to $\C$.}
\label{fig:ex}
\end{figure}

Now we give a formal description of the problem. A \emph{\fd{}} is a node-labelled DAG containing a source node $s$ and sink node $t$, and where all other nodes are labelled with a criterion.
Given a set $\mathcal{T}$ of sequences of states and set of criteria $\C$, the goal is to construct a \fd{} with a minimum number of nodes, such that a segmentation of each sequence of states in $\mathcal{T}$ is represented, that is, included as an $s$--$t$ path, in the \fd{}.
Furthermore (when criteria depend on multiple state sequences, e.g.\ $C_7$ in Fig.~\ref{fig:ex}) we require that the segmentations represented in the \fd{} are consistent, i.e.\ can be jointly realized.
The \emph{\FD{} problem} thus requires the segmentations of each sequence of states and the minimal \fd{} of the segmentations to be computed. It can be stated as:

\begin{problem}
    \FD{} (FD) \\
\noindent {\bf Instance:} A set of sequences of states $\T = \{\tau_1, \ldots, \tau_m\}$, each of length at most $n$, a set of criteria $\C = \{C_1,\ldots,C_k\}$ and an integer $\lambda > 2$.\\
\noindent {\bf Question:} Is there a \fd{} $\F$ with $\leq\lambda$ nodes, such that for each $\tau_i\in\T$, there exists a segmentation according to $\C$ which appears as an $s$--$t$ path in $\F$?
\end{problem}

Even the small example above shows that there can be considerable space savings by representing a set of state sequences as a \fd{}. This is not a lossless representation and comes at a cost.
The \fd{} represents the sequence of flow between states, however, the information about an individual sequence of states is lost. As we will argue in Section~\ref{sec:experiments}, paths representing many segments in the obtained \fd{}s show interesting patterns. 
We will give two examples. First we consider segmenting the morphology of formations of a defensive line of football players during a match (Fig.~\ref{fig:flow_diagram_defence_formation}). The obtained \fd{} provides an intuitive summary of these formations. 
The second example models attacking possessions as state sequences. The summary given by the \fd{} gives intuitive information about differences in attacking tactics. 

\subsubsection*{Properties of Criteria.}
The efficiency of the algorithms will depend on properties of the criteria on which the segmentations are based. Here we consider four cases: (i) general criteria without restrictions; (ii) monotone decreasing and independent criteria; (iii) monotone decreasing and dependent criteria; and (iv) fixed criteria. To illustrate the properties we will again use the example in Fig.~\ref{fig:ex}.

A criterion $C$ is \emph{monotone decreasing}~\cite{bdks-stfa-11} for a given sequence of states $\tau$ that fulfils $C$, if all subsequences of $\tau$ also fulfil $C$. For example, if $C_4$ is fulfilled by a sequence $\tau$ then any subsequence $\tau'$ of $\tau$ will also fulfil $C_4$. This is in contrast to criterion $C_5$ which is not monotone decreasing.

A criterion $C$ is \emph{independent} if checking whether a subsequence $\tau'$ of a sequence $\tau_i \in \mathcal{T}$ fulfils $C$ can be achieved without reference to any other sequences $\tau_j \in \mathcal{T}, i \neq j$.
Conversely, $C$ is \emph{dependent} if checking that a subsequence $\tau'$ of $\tau_i$ requires reference to other state sequences in $\mathcal{T}$.
In the above example $C_4$ is an example of an independent criterion while $C_7$ is a dependent criterion since it requires that at least two objects fulfil the criterion at the same time.


\subsubsection*{Related work.}
To the best of our knowledge compactly representing sequences of states as \fd{}s has not been considered before. The only related work we are aware of comes from the area of trajectory analysis. Spatial trajectories are a special case of state sequences. A spatial trajectory describes the movement of an object through space over time, where the states are location points, which may also include additional information such as heading, speed, and temperature. For a single trajectory a common way to obtain a compact representation is \emph{simplification}~\cite{cwt-stdrd-06}. Trajectory simplification asks to determine a subset of the data that represents the trajectory well in terms of the location over time. 
If the focus is on characteristics other than the location, then \emph{segmentation}~\cite{abbkkw-tssc-14,adkls-stnmc-13,bdks-stfa-11} is used to partition a trajectory into a small number of subtrajectories, where each subtrajectory is homogeneous with respect to some characteristic. 
This allows a trajectory to be compactly represented as a sequence of characteristics.

For multiple trajectories other techniques apply. A large set of trajectories might contain very unrelated trajectories, hence \emph{clustering} may be used. Clustering on complete trajectories will not represent information about interesting parts of trajectories; for this clustering on subtrajectories is needed~\cite{bbgll-dcpcs-11,hjzs-scabs-11}. 
A set of trajectories that forms different groups over time may be captured by a \emph{grouping structure}~\cite{bbkss-tgs-13}. These approaches also focus on location over time.

For the special case of spatial trajectories, a \fd{} can be illustrated by a simple example: trajectories of migrating geese, see~\cite{bkk-stmb-12}. 
The individual trajectories can be segmented into phases of activities such as directed flight, foraging and stop overs. 
This results in a \fd{} containing a path for the segmentation of each trajectory. 
More complex criteria can be imagined that depend on a group of geese, or frequent visits to the same area, resulting in complex state sequences that are hard to analyze without computational tools.

\subsubsection*{Organization}
In Section~\ref{sec:Algorithms} we present algorithms for the \FD{} problem using criteria with the properties described above.
These algorithms only run in polynomial time if the number of state sequences $m$ is constant. 
Below we observe that this is essentially the best we can hope for by showing that the problem is $W[1]$-hard.
Both theorems are proved in Section~\ref{sec:Hardness}.
Unless $W[1]=FPT$, this rules out the existence of algorithms with time complexity of $O(f(m) \cdot (nk)^c)$ for some constant~$c$, where $m,n$ and $k$ are the number of state sequences, the length of the state sequences and the number of criteria, respectively. 
To obtain \fd{}s for larger groups of state sequences we propose two heuristics for the problem in Section~\ref{sec:Algorithms}. We experimentally evaluate the algorithms and heuristics in Section~\ref{sec:experiments}.

\section{Hardness Results}
\label{sec:Hardness}

In this section, the following hardness results are proven.

\begin{theorem}\label{thm:hard}
  The FD problem is NP-hard. This even holds when only two criteria are used or when the length of every state sequence is $2$. Furthermore, for any $0 < c < 1/4$, the FD problem cannot be approximated within factor of $c \log m$ in polynomial time unless $NP\subset DTIME(m^{\operatorname{polylog} m})$.
\end{theorem}

Also for bounded $m$ the running times of our algorithms is rather high. Again, we can show that there are good reasons for this.

\begin{theorem}\label{thm:w1}
  The FD problem parameterized in the number of state sequences is $W[1]$-hard even when the number of criteria is constant.
\end{theorem}

To obtain the stated results we will perform two reductions; one from the Shortest Common Supersequence problem and one from the Set Cover problem.

\subsection{Reduction from SCS}

\begin{problem}
  Shortest Common Supersequence (SCS) \\
\noindent {\bf Instance:} A set of strings $R = \{r_1, r_2, \ldots, r_k\}$ over an alphabet $\Sigma$, a positive
integer $\lambda$.\\
\noindent {\bf Question:} Does there exist a string $s \subset \Sigma^*$ of length at most $\lambda$, that
is a supersequence of each string in $R$?
\end{problem}

The SCS problem has been extensively studied over the last 30 years (see~\cite{fi-aascs-95} and references therein). Several hardness results are known, we will use the following two.

\begin{lemma}[Pietrzak~\cite{p-pcfas-03}]\label{fact:scs-hard-1}
 The Shortest Common Supersequence problem parameterized in the number of Strings is $W[1]$ hard even when the alphabet has constant size.
\end{lemma}

\begin{lemma}[R\"aih\"a and E. Ukkonen~\cite{ru-scspb-81}]\label{fact:scs-hard-2}
 The Shortest Common Supersequence problem over a binary alphabet is NP-complete.
\end{lemma}


Given an instance $I=(R=\{r_1, \ldots, r_m\}, \Sigma)$ of SCS construct an instance of FD as follows.
Each character $c_l$ in the alphabet $\Sigma$ corresponds to a criterion $c_l$. Each string $r_i$ corresponds to a state sequence $T_i$, where $T_i[j]=c_{r_i[j]}$. Thus at any step $T_i$ fulfils exactly one criterion.

An algorithm for FD given an instance outputs a \fd{} $F$ of size~$f$. Given $F$ one can compute a linear sequence $b$ of the vertices of $F$ using topological sort, as shown in Fig.~\ref{subfig:reduction-scs}. The linear sequence $b$ has $f-2$ vertices (omitting the start and end state of $F$) and it is a supersequence of each string in $R$. It follows that the size of $F$ is $\lambda$ if the number of characters in the SCS of $I$ has length $\lambda-2$. Note that $F$ contains a linear sequence of vertices (after topological sort), which correspondence to a supersequence, and a set of directed edges. Consequently a solution for the FD problem can easily be transformed to a solution for the SCS problem but not vice versa.

From the above reduction, together with Lemmas~\ref{fact:scs-hard-1}-\ref{fact:scs-hard-2}, we obtain Theorem~\ref{thm:w1} and the following. 


\begin{lemma}\label{lem:np2crit}
  The FD problem is NP-hard even for two criteria.
\end{lemma}

\begin{figure}[tb]
\centering
  \begin{subfigure}[b]{0.3\linewidth}
    \centering
    \includegraphics[page=1]{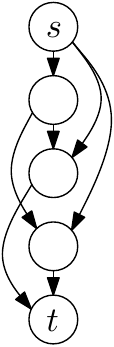}
    \caption{}
    \label{subfig:reduction-scs}
  \end{subfigure}
  \begin{subfigure}[b]{0.5\linewidth}
    \centering
    \includegraphics[page=2]{Reduction}
    \caption{}
    \label{subfig:reduction-sc}
  \end{subfigure}
\caption{Examples of \fd{}s produced by the reductions: (\subref{subfig:reduction-scs}) From \textsc{Shortest Common Supersequence}. (\subref{subfig:reduction-sc}) From \textsc{Set Cover}.}
\label{fig:Reduction}
\end{figure}

\subsection{Reduction from Set Cover}

\begin{problem}
  Set Cover (SC) \\
\noindent {\bf Instance:} A set of elements $E = \{e_1, e_2, \ldots , e_m\}$, a set of $n$ subsets
of $E$, $S = \{S_1, S_2, \ldots , S_n\}$ and a positive integer $\lambda$.\\
\noindent {\bf Question:} Does there exist set of $\lambda$ items in $S$ whose union equals $E$?
\end{problem}

Set Cover is well known to be NP-hard, and also hard to approximate:

\begin{lemma}[Lund and Yannakakis~\cite{ly-hamp-93}] \label{fact:sc-hard-1} 
 For any $0 < c < 1/4$, the Set Covering problem cannot be approximated within factor of $c \log m$ in polynomial time unless $NP \subset DTIME(m^{polylog m})$.
\end{lemma}

Given an instance $I=(E = \{e_1, e_2, \ldots , e_m\}, S = \{S_1, S_2, \ldots , S_n\})$ of Set Cover construct an instance of FD as follows.
Each item $e_i$ in $E$ corresponds to a state sequence $T_i$ of length two. Each subset $S_j$ corresponds to a criterion $C_j$. If a $S_j$ contains $e_i$ then the whole state sequence $T_i$ fulfils criterion $C_j$.

An algorithm for FD given the new instance outputs a \fd{} $F$ of size $f$. The output $F$ is depicted in Fig.~\ref{subfig:reduction-sc}. Given $F$ the interior vertices of $F$ corresponds to a set of subsets in $S$ whose union is $E$. The diagram $F$ has $f$ vertices if and only there is $f-2$ subsets in $S$ that forms a Set Cover of $E$.

We obtain Theorem~\ref{thm:hard} from the above reduction, together with Lemma~\ref{fact:sc-hard-1}.

%

\section{Algorithms}
\label{sec:Algorithms}

In this section, we present algorithms that compute a smallest \fd{} representing a set of $m$ state sequences of length $n$ for a set of $k$ criteria.
First, we present an algorithm for the general case, followed by more efficient algorithms for the case of monotone increasing and independent criteria, the case of monotone increasing and dependent criteria, and then two heuristic algorithms.

\subsection{General criteria}

Next, we present a dynamic programming algorithm for finding a smallest \fd{}.
Recall that a node $v$ in the \fd{} represents a criterion $C_j$ that is fulfilled by a contiguous 
segment in some of the state sequences. 
Let $\tau[i,j]$, $i\leq j$, denote the subsequence of $\tau$ starting at the $i$th state of $\tau$ and ending at the $j$th state, where $\tau[i,i]$ is the empty sequence. Construct an $(n+1)^m$ grid of vertices, where a vertex with coordinates $(x_1, \ldots, x_m)$, $0\leq x_1, \ldots, x_m \leq n$, represents $(\tau_1[0,x_1], \dots, \tau_m[0,x_m])$. 
Construct a \emph{prefix graph} $G$ 
as follows:

There is an edge between two vertices $v = (x_1,\dots, x_m)$ and $v' = (x'_1,\dots, x'_m)$, labeled by some criterion $C_j$, if and only if, for every $i$, $1\leq i \leq m$, one of the following two conditions is fulfilled:
(1) $x_i=x'_i$, or (2) all remaining $\tau_i[x_i+1,x'_i]$ jointly fulfil $C_j$.
Consider the edge between $(x_1,x_2)=(1,0)$ and $(x'_1,x'_2)=(1,1)$ in Fig.~\ref{fig:algo1}(b). Here $x_1=x'_1$ and $\tau_2[x_2+1,x'_2]$ fulfils $C_2$.

Finally, define $v_s$ to be the vertex in $G$ with coordinates $(0,\ldots, 0)$ and add an additional vertex $v_t$ outside the grid, which has an incoming edge from $(n, \ldots, n)$. This completes the construction of the prefix graph $G$.

\begin{figure}[tb]
    \includegraphics[width=\textwidth]{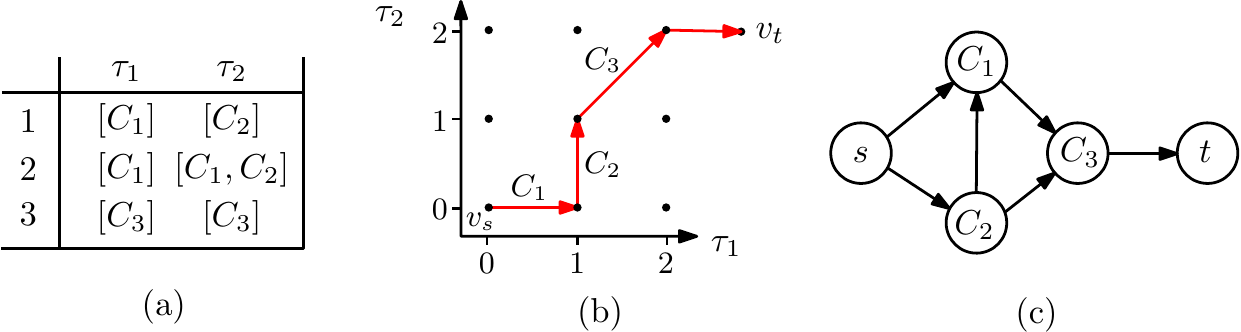}
    \caption{(a) A segmentation of $\T=\{\tau_1,\tau_2\}$ according to $\C=\{C_1,C_2,C_3\}$. (b) The prefix graph $G$ of the segmentation, omitting all but four of the edges. (c) The resulting \fd{} generated from the highlighted path in the prefix graph.}
    \label{fig:algo1}
\end{figure}

Now, a path in $G$ from $v_s$ to a vertex $v$ represents a valid segmentation of some prefix of each state sequence, and defines a \fd{} that describes these segmentations in the following way:
the empty path represents the \fd{} consisting only of the start node $s$.
Every edge of the path adds one new node to the \fd{}, labeled by the criterion that the segments fulfil. 
Additionally, for each node the \fd{} contains an edge from every node representing a previous segment, or from $s$ if the node is the first in a segmentation.
For a path leading from $v_s$ to $v_t$, the target node $t$ is added to the \fd{}, together with its incoming edges.
This ensures that the \fd{} represents valid segmentations and that each node represents at least one segment.
An example of this construction is shown in Fig.~\ref{fig:algo1}.

Hence the length of a path (where length is the number of edges on the path) equals the number of nodes of the corresponding \fd{}, excluding $s$ and~$t$.
Thus, we find an optimal \fd{} by finding a shortest $v_s$--$v_t$ path in $G$.
\begin{restatable}{lemma}{sfdgeneralsp}
\label{lem:sfd-sp}
A smallest \fd{} for a given set of state sequences is represented by a shortest $v_s$--$v_t$ path in $G$.
\end{restatable}

\begin{proof}
We show that every $v_s$--$v_t$ path $P$ in $G$ represents a valid \fd{} $F$, with the path length equal to the \fd{}'s cost, and vice versa.
Thus, a shortest path represents a minimal valid \fd{} for the given state sequences.

Let $P := (v_s =: v_1, v_2, \dots, v_\ell := v_t)$ be a $v_s$--$v_t$ path of length $\ell-1$ in $G$.
As described in the text, every $v_s$--$v_t$ path in $G$ represents a valid \fd{}, and every vertex visited by the path contributes exactly one node to the \fd{}.
Thus, $P$ represents a valid \fd{} with exactly $\ell$ nodes.

For the other direction, let $F$ be a valid \fd{} of a set of state sequences $\{T_1,\dots,T_m\}$, each of length $n$.
That is, there are segmentations $\S_1,\dots,\S_m$ of the state sequences such that every segmentation is represented in $F$ in the following way:
assume the nodes of $F$ are $\{s=:f_1,f_2,\dots,f_\ell:=t\}$ according to some topological sorting.
Let $\S_j$ consist of the segments $s_{j,1},\dots,s_{j,\sigma_j}$, where $\sigma_j$ is the number of segments in $\S_j$.
Then there exists a path $(s=:f_{j,0},f_{j,1},\dots,f_{j,\sigma_j},t)$ in $F$ such that each segment $s_{j,i}$ fulfils the criterion $C(f_{j,i})$ associated with $f_{j,i}$.

Let $b_{j,i}$ be the index in $T_j$ at which $s_{j,i}$ ends, for $1\leq i\leq \sigma_j$, and let $b_{j,0} := 1$.
Since $\S_j$ is a segmentation of $T_j$, $b_{j,\sigma_j}=n$.
Let $F_\lambda$ be the subdiagram of $F$ induced by $(f_1,\dots,f_\lambda)$, for $1\leq\lambda\leq\ell$.
We define $x_{j,\lambda} := \max \{b_{j,i} \mid f_{j,i}\in\{f_1,\dots,f_\lambda\}\}$.
We show inductively that for each $\lambda \in \{1,2,\dots,\ell\}$, $G$ contains a path from $v_s$ to the vertex $v_\lambda := (x_{1,\lambda},\dots,x_{m,\lambda})$ with length $\lambda-1$, i.e. the number of nodes in $F_\lambda$ excluding $s$.
\begin{itemize}
\item Base case $\lambda = 1$:
Note that $v_1 = v_s$, and thus there is a path of length $ \lambda-1 = 0$ from $v_s$ to $v_1$.

\item Induction step: The node $f_{\lambda+1}$ represents the segments \[\{T_j[x_{j,\lambda},x_{j,\lambda+1}] \mid 1\leq j\leq m \wedge x_{j,\lambda} \neq x_{j,\lambda+1}\}.\]
Since the \fd{} is valid, these segments fulfil the criterion $C(f_{\lambda+1})$, and thus $G$ contains an edge from $v_\lambda$ to $v_{\lambda+1}$.
Since a path from $v_s$ to $v_\lambda$ of length $\lambda$ exists by the induction hypothesis, there is a path from $v_s$ to $v_{\lambda+1}$ of length $\lambda+1$.
\end{itemize}

For every state sequence $T_j$, there exists an index $\varphi_j \in \{1,\dots,\ell-1\}$ such that $x_{j,\lambda}=n$ for all $\lambda\geq\varphi_j$.
Thus, $v_{\ell-1} = (n,n,\dots,n)$ and $G$ contains an edge from $v_{\ell-1}$ to $v_\ell = v_t$.
So, there is a path from $v_s$ to $v_t$ of length $\ell-1$.
\qed
\end{proof}
Recall that $G$ has $(n+1)^m$ vertices. 
Each vertex has $O(k (n+1)^m)$ outgoing edges,
thus, $G$ has $O(k(n+1)^{2m})$ edges in total. To decide if an edge is present in $G$, check if the nonempty segments the edge represents fulfil the criterion. Thus, we need to perform $O(k(n+1)^{2m})$ of these checks.
There are $m$ segments of length at most $n$, and we assume the cost for checking this is $T(m,n)$.
Thus, the cost of constructing $G$ is $O(k(n+1)^{2m}\cdot T(m,n))$, and finding the shortest path requires $O(k(n+1)^{2m})$ time.


\begin{theorem}
The algorithm described above computes a smallest \fd{} for a set of $m$ state sequences, each of length at most $n$, and $k$ criteria in $O((n+1)^{2m} k\cdot T(m,n))$ time, where $T(m,n)$ is the time required to check if a set of $m$ subsequences of length at most $n$ fulfils a criterion.
\end{theorem}

\subsection{Monotone decreasing and independent criteria}
If all criteria are decreasing monotone and independent, we can use ideas similar to those presented in~\cite{bdks-stfa-11} to avoid constructing the full graph. 
%
From a given vertex with coordinates $(x_1,\dots,x_m)$, we can greedily move as far as possible along the sequences, since the monotonicity guarantees that this never leads to a solution that is worse than one that represents shorter segments.
For a given criterion $C_j$, we can compute for each $\tau_i$ independently the maximum $x'_i$ such that $\tau_i[x_i+1,x'_i]$ fulfils $C_j$.
This produces coordinates $(x'_1,\dots,x'_m)$ for a new vertex,
which is the optimal next vertex using $C_j$.
By considering all criteria we obtain $k$ new vertices. 
However, unlike the case with a single state sequence, there is not necessarily one vertex 
that is better than all others (i.e. largest ending position), since there is no total order on the vertices.
Instead, we consider all vertices that are not dominated by another vertex, where a vertex $p$ dominates a vertex $p'$ if each coordinate of $p$ is at least as large as the corresponding coordinate of $p'$, and at least one of $p$'s coordinates is larger.

Let $V_i$ be the set of vertices of $G$ that are reachable from $v_s$ in exactly $i$ steps, and define $M(V) := \{v\in V\mid \text{no vertex $u\in V$ dominates $v$}\}$ to be the set of \emph{maximal vertices} of a vertex set $V$.
Then a shortest $v_s$--$v_t$ path through $G$ can be computed by iteratively computing $M(V_i)$ for increasing $i$, until a value of $i$ is found for which $v_t\in M(V_i)$.
Observe that $|M(V)| = O((n+1)^{m-1})$ for any set $V$ of vertices in the graph.
Also note that $V_0 = M(V_0) = v_s$.

\begin{restatable}{lemma}{MDIa}
\label{lem:Di-step}
For each $i\in\{1,\dots,\ell-1\}$, every vertex in $M(V_i)$ is reachable in one step from a vertex in $M(V_{i-1})$. Here, $\ell$ is the distance from $v_s$ to $v_t$.
\end{restatable}
\begin{proof}
Assume there exists a vertex $v \in M(V_i)$ that has no edge from a vertex in $M(V_{i-1})$.
Since $v\in M(V_i)$, $v$ is also contained in $V_i$, and thus its distance from $v_s$ is $i$.
Thus, there must be a vertex $v'$ at distance $i-1$ from $v_s$, i.e. $v'\in V_{i-1}$, that has an edge to $v$ representing a criterion $C_j$.
By assumption, $v'$ is not contained in $M(V_{i-1})$, and thus there is a vertex $v'' \in M(V_{i-1})$ that dominates $v'$.
But then, by the monotonicity of $C_j$, there must be a vertex reachable from $v''$ that is identical to $v$ or dominates $v$. Both cases lead to a contradiction.
\qed
\end{proof}

$M(V_i)$ is computed by computing the farthest reachable vertex for each $v\in M(V_{i-1})$ and criterion, thus yielding a set $D_i$ of $O((n+1)^{m-1}k)$ vertices.
This set contains $M(V_i)$ by Lemma~\ref{lem:Di-step}, so we now need to remove all vertices that are dominated by some other vertex in the set to obtain $M(V_i)$.

We find $M(V_i)$ using a copy of $G$.
Each vertex may be marked as being in $D_i$ or dominated by a vertex in $D_i$.
We process the vertices of $D_i$ in arbitrary order.
For a vertex $v$, if it is not yet marked, we mark it as being in $D_i$.
When a vertex is newly marked, we mark its $\leq m$ immediate neighbours dominated by it as being dominated.
After processing all vertices, the grid is scanned for the vertices still marked as being in $D_i$.
These vertices are exactly $M(V_i)$.

When computing $M(V_i)$, $O((n+1)^{m-1}k)$ vertices need to be considered, and the maximum distance from $v_s$ to $v_t$ is $m(n+1)$, so the algorithm considers $O(mk(n+1)^m)$ vertices.
We improve this bound by a factor $m$ using the following: 
\begin{restatable}{lemma}{MDIb}
\label{lem:ind-mondec:Di_size}
The total size of all $D_i$, for $0\leq i\leq\ell-1$, is $O(k(n+1)^m)$.
\end{restatable}
\begin{proof}
If a vertex $v$ appears in $M(V_i)$ for some $i\in\{0,\dots,\ell-1\}$, it generates vertices for $D_{i+1}$ that dominate $v$, and thus $v\not\in M(V_{i+j})$ for any $j>0$.
So, each of the $n^m$ vertices appears in at most one $M(V_i)$ and generates $k$ candidate vertices for $D_{i+1}$ (not all unique).
Hence the total size of all $D_i$ is $O(kn^m)$.
\qed
\end{proof}

Using this result, we compute all $M(V_i)$ in $O((k+m)(n+1)^m)$ time, since $O(k(n+1)^m)$ vertices are marked directly, and each of the $(n+1)^m$ vertices is checked at most $m$ times when a direct successor is marked.
One copy of the grid can be reused for each $M(V_i)$, since each vertex of $D_{i+1}$ dominates at least one vertex of $M(V_i)$ and is thus not yet marked while processing $D_j$ for any $j\leq i$.

Since the criteria are independent, the farthest reachable point for a given starting point and criterion can be precomputed for each state sequence separately.
Using the monotonicity we can traverse each state sequence once per criterion and thus need to test only $O(nmk)$ times whether a subsequence fulfils a criterion.

\begin{theorem}
The algorithm described above computes a smallest \fd{} for $m$ state sequences of length $n$ with $k$ independent and monotone decreasing criteria in $O({mnk\cdot T(1,n)} + (k+m)(n+1)^m)$ time, where $T(1, n)$ is the time required to check if a subsequence of length at most $n$ fulfils a criterion.
\end{theorem}

\subsection{Monotone decreasing and dependent criteria}
For monotone decreasing and dependent criteria, we can use a similar approach to that described above, however, for a given start vertex $v$ and criterion $C$, there is not a single vertex $v'$ that dominates all vertices reachable from $v$ using this criterion. Instead there may be $\Theta((n+1)^{m-1})$ maximal reachable vertices from $v$ for criterion $C$. The maximal vertices can be found by testing $O((n+1)^{m-1})$ vertices on or near the upper envelope of the reachable vertices 
in $O((n+1)^{m-1}\cdot T(m,n))$ time. Using a similar reasoning as in Lemma~\ref{lem:ind-mondec:Di_size}, we can show that the total size of all $D_i$ ($0\leq i\leq\ell-1$) is $O(k(n+1)^{2m-1})$, which gives:

\begin{theorem}
The algorithm from the previous section computes a smallest \fd{} for $m$ state sequences of length $n$ with $k$ monotone decreasing criteria in $O(k(n+1)^{2m-1}\cdot T(m,n) + m(n+1)^m)$ time, where $T(m, n)$ is the time required to check if a set of $m$ subsequences of length at most $n$ fulfils a criterion.
\end{theorem}

\subsection{Heuristics}
The hardness results presented in the introduction
indicate that it is unlikely that the performance of the algorithms will be acceptable in practical situations, except for very small inputs. As such, we investigated heuristics that may produce usable results that can be computed in reasonable time.

For monotone decreasing and independent criteria, the heuristics we consider are based on the observation that by limiting $V_i$, the vertices that are reachable from $v_s$ in $i$ steps, to a fixed size, the complexity of the algorithm can be controlled.  Given that every path in a prefix graph represents a valid \fd{}, any path chosen in the prefix graph will be valid, though not necessarily optimal. In the worst case, a vertex that advances along a single state sequence a single time-step (i.e.\ advancing only one state) will be selected, and for each vertex, all $k$ criteria must be evaluated, so $O(kmn)$ vertices may be processed by the algorithm. We consider two strategies for selecting the vertices in $V_i$ to retain:

\noindent (1) For each vertex in $V_i$, determine the number of state sequences that are advanced in step $i$ and retain the top $q$ vertices [\emph{sequence heuristic}].

\noindent (2) For each vertex in $V_i$, determine the number of time-steps that are advanced in all state sequences in step $i$ and retain the top $q$ vertices [\emph{time-step heuristic}].

In our experiments we use $q=1$ since any larger value would immediately give an exponential worst-case running time.




\section{Experiments} \label{sec:experiments}

%
The objectives of the experiments were twofold: to determine whether compact and useful \fd{}s could be produced in real application scenarios; and to empirically investigate the performance of the algorithms on inputs of varying sizes.
We implemented the algorithms described in Section~\ref{sec:Algorithms} using the Python programming language.
For the first objective, we considered the application of {\fd}s to practical problems in football analysis in order to evaluate their usefulness.
For the second objective, the algorithms were run on generated datasets of varying sizes to investigate the impact of different parameterisations on the computation time required to produce the \fd{} and the complexity of the \fd{} produced.

\subsection{Tactical Analysis in Football}
\label{sub:football-analysis}

Sports teams will apply tactics to improve their performance, and computational methods to detect, analyse and represent tactics have been the subject of several recent research efforts~\cite{bialkowski-2014,gudmundsson-2014,lucey-2013,vanhaaren-2015,wang-2015,wei-2013}. 
Two manifestations of team tactics are in persistent and repeated occurrence of spatial formations of players, and in \emph{plays} -- a coordinated sequence of actions by players.
We posited that \fd{}s would be a useful tool for compactly representing both these manifestations, and we describe the approaches used in this section.

The input for the experiments is a database containing player trajectory and event data from four home matches of the Arsenal Football Club from the $2007/08$ season, provided by Prozone Sports Limited~\cite{prozone}. 
For each player and match, there is a trajectory comprising a sequence of timestamped location points in the plane, sampled at \SI{10}{\Hz} and accurate to \SI{10}{cm}. 
In addition, for each match, there is a log of all the match events, comprising the timestamp and location of each event.

\subsubsection{Defensive Formations.}
\label{sub:defensive-formations}

The spatial formations of players in football matches are known to characterize a team's tactics~\cite{bialkowski-2014a}, and a compact representation of how formations change over time would be a useful tool for analysis.
We investigated whether a \fd{} could provide such a compact representation of the defensive formation of a team, specifically to show how the formation evolves during a phase of play.
The trajectories of the four defensive players were re-sampled at one-second intervals and used to compute a sequence of formation states which were then segmented to model the formation.



The criteria were derived from those presented by Kim~et~al.~\cite{kim2011}. The angles between pairs of adjacent players (along the defence line) were used to compute the formation criteria, see Fig~\ref{fig:formation_processing}. 
We extended this scheme to allow multiple criteria to be applied where the angle between pairs of players is close to $10^{\circ}$.
The reason for this was to facilitate compact results by allowing for smoothing of small variations in contiguous time-steps.
The criteria applied to each state is a triple $(x_1, x_2, x_3)$, computed as follows.
Given two player positions $p$ and $q$, let $\angle p q$ be the angle between $p$ and $q$ relative to the goal-line.
Let $R(-1) = [-90^{\circ}, -5^{\circ})$, $R(0) = (-15^{\circ}, +15^{\circ})$, and $R(+1) = (+5^{\circ}, +90^{\circ}]$ be three angular ranges.
The positions 
of the four defenders satisfy the criteria (and thus have the formation) $(x_1, x_2, x_3)$ if $\angle p_i p_{i+1} \in R(x_i)$ for all $i \in \{1,2,3\}$.


\begin{figure}[tb]
    \begin{center}
        \includegraphics[width=0.99\textwidth]{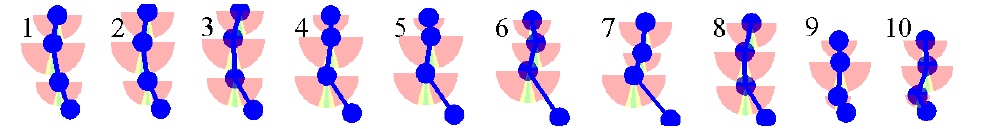}
        \includegraphics[width=0.8\textwidth]{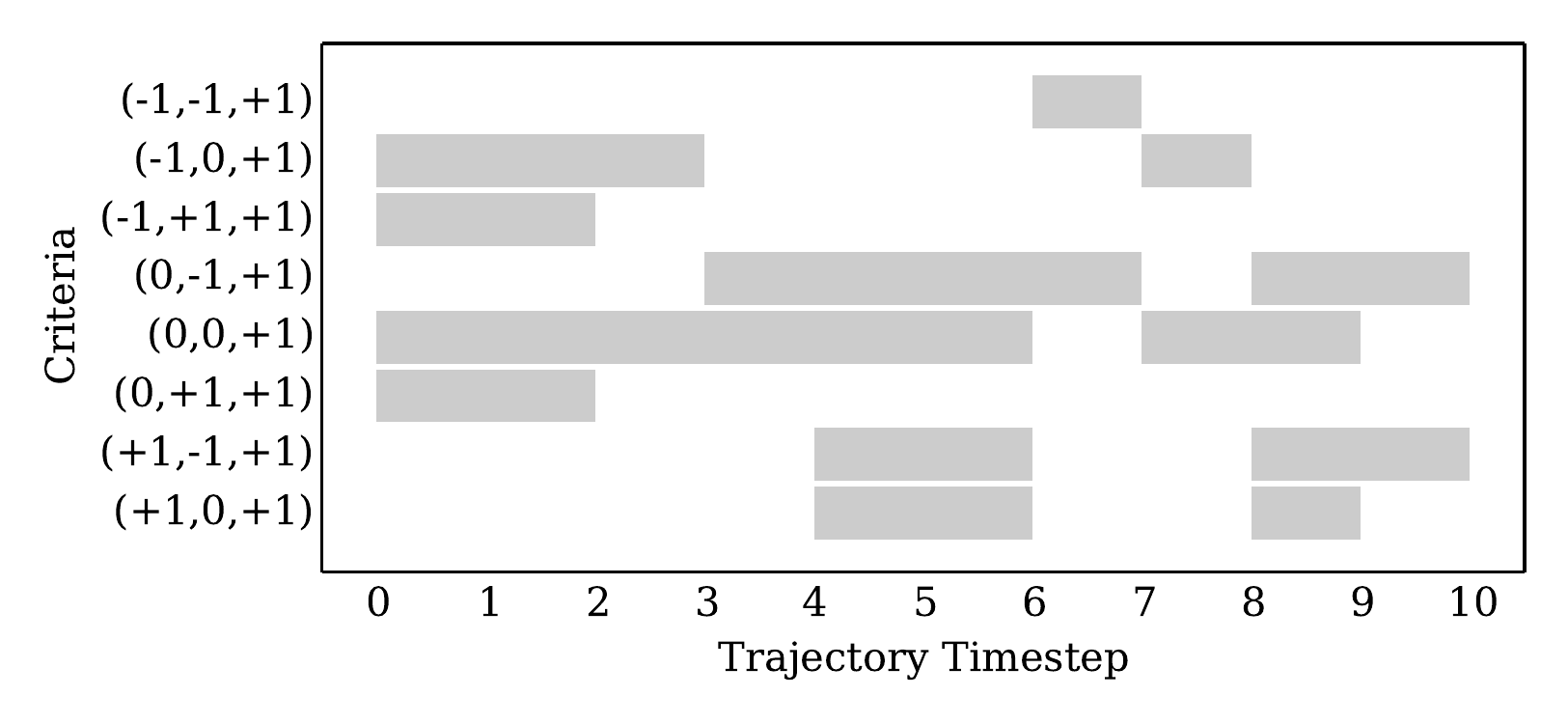}
        \caption{Segmentation of a single state sequence. The formation state sequence is used to compute the segmentation representation, where segments corresponding to criteria span the state sequence \emph{(bottom)}. The representation of this state sequence in the movement \fd{} is shaded in Fig.~\ref{fig:flow_diagram_defence_formation}.}
        \label{fig:formation_processing}
    \end{center}
\end{figure}

The criteria in this experiment were monotone decreasing and independent, and we ran the corresponding algorithm using randomly selected sets of the state sequences as input. The size $m$ of the input was increased until the running time exceeded a threshold of $6$ hours.
The algorithm successfully processed up to $m=12$ state sequences, having a total of $112$ assigned segments.
The resulting \fd{}, Fig.~\ref{fig:flow_diagram_defence_formation}, has a total complexity of $12$ nodes and $27$ edges. 

\begin{figure}[tb]
    \begin{center}
    \includegraphics[width=0.8\textwidth]{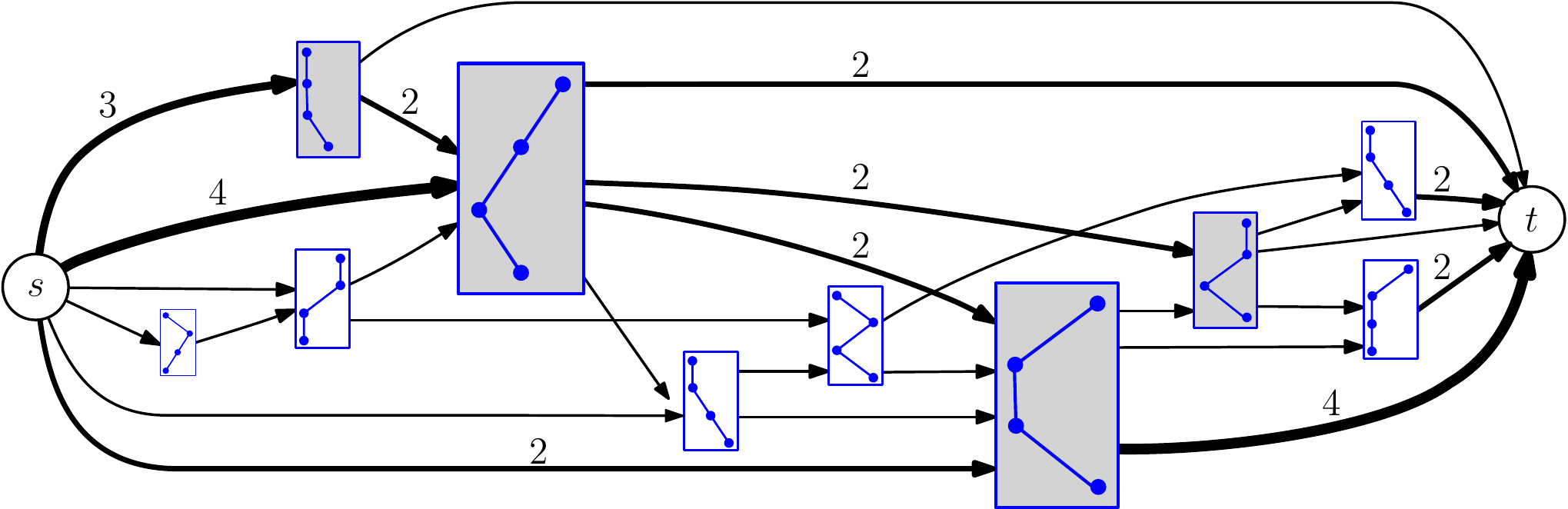}
    \caption{\Fd{} for formation morphologies of twelve defensive possessions. The shaded nodes are the segmentation of the state sequence in Fig.~\ref{fig:formation_processing}.}
    \label{fig:flow_diagram_defence_formation}
    \end{center}
\end{figure}

We believe that the \fd{} provides an intuitive summary of the defensive formation, and several observations are apparent.
There appears to be a preference amongst the teams for the right-back to position himself in advance of the right centre-half (i.e. the third component of the triple is $+1$).
Furthermore, the $(0,0,0)$ triple, corresponding to a ``flat back four'' is not present in the diagram. This is typically considered the optimal formation for teams that utilise the offside trap, and thus may suggest that the defences here are not employing this tactic.
These observations were apparent to the authors as laymen, and we would expect that a domain expert would be able to extract further useful insights from the \fd{}s.



\subsubsection{Attacking Plays.}
\label{sub:attacking-plays}

During a football match, the team in possession of the ball is attempting to reach a position where they can take a shot at goal.
Teams will typically use a variety of tactics to achieve such a position, e.g. teams can vary the intensity of an attack by pushing forward, moving laterally, making long passes, or retreating. 
We modelled attacking possessions as state sequences, segmented according to criteria representing the attacking intensity and tactics employed, and computed \fd{}s for the possessions.
In particular, we were interested in determining whether differences in tactics employed by teams when playing at home or away~\cite{bialkowski-2014} are apparent in the \fd{}s.

We focus on \emph{ball events}, where a player touches the ball, e.g. passes, touches, dribbles, headers, and shots at goal.
The event sequence for each match was divided into sub-sequences where a single team was in possession, and then filtered to include only those that end with a shot at goal. 

We defined criteria that characterised the movement of the ball - relative to the goal the team is attacking - between event states in the possession sequence. 
The applied criteria are defined as follows.
Let $x_i$, $y_i$, $t_i$ be the $x$-coordinate in metres, $y$-coordinate in metres and time-stamp in seconds, respectively, for event $i$. The velocity of the ball in the $x$-direction at event $i$, which is in the direction of the goal, is thus $x_v = (x_{i+1} - x_{i}) / (t_{i+i} - t_{i})$ in $m/s$.
The velocity $y_v$ of the ball in the $y$ direction is computed in a similar fashion, and together are used to specify the following criteria.

\begin{itemize}
    \item \emph{Backward movement (BM):} $x_v < 1$, a sub-sequence of passes or touches that move in a defensive direction.
    \item \emph{Lateral movement (LM):} $-5 < x_v < 5$, passes or touches that move in a lateral direction.
    \item \emph{Forward movement (FM):} $-1 < x_v < 12$, passes or touches that move in an attacking direction, at a velocity in the range achievable by humans, i.e. to approximately $10 m/s$.
    \item \emph{Fast forward movement (FFM):} $8 < x_v$, passes or touches moving in an attacking direction at a velocity generally in exceeds of maximum human velocity.
    \item \emph{Long ball (LB):} a pass travelling $30m$ in the attacking direction.
    \item \emph{Cross-field ball (CFB):} a pass travelling $20m$ in the cross-field direction, and that has angle in range $[80,100]$ or $[-80,-100]$.
    \item \emph{Shot resulting in goal (SG):} a successful shot resulting in a goal.
    \item \emph{Shot not resulting in goal (SNG):} an unsuccessful shot that does not produce a goal.
\end{itemize}


\begin{figure}[tb] 
    \centering
    \includegraphics[scale=0.23]{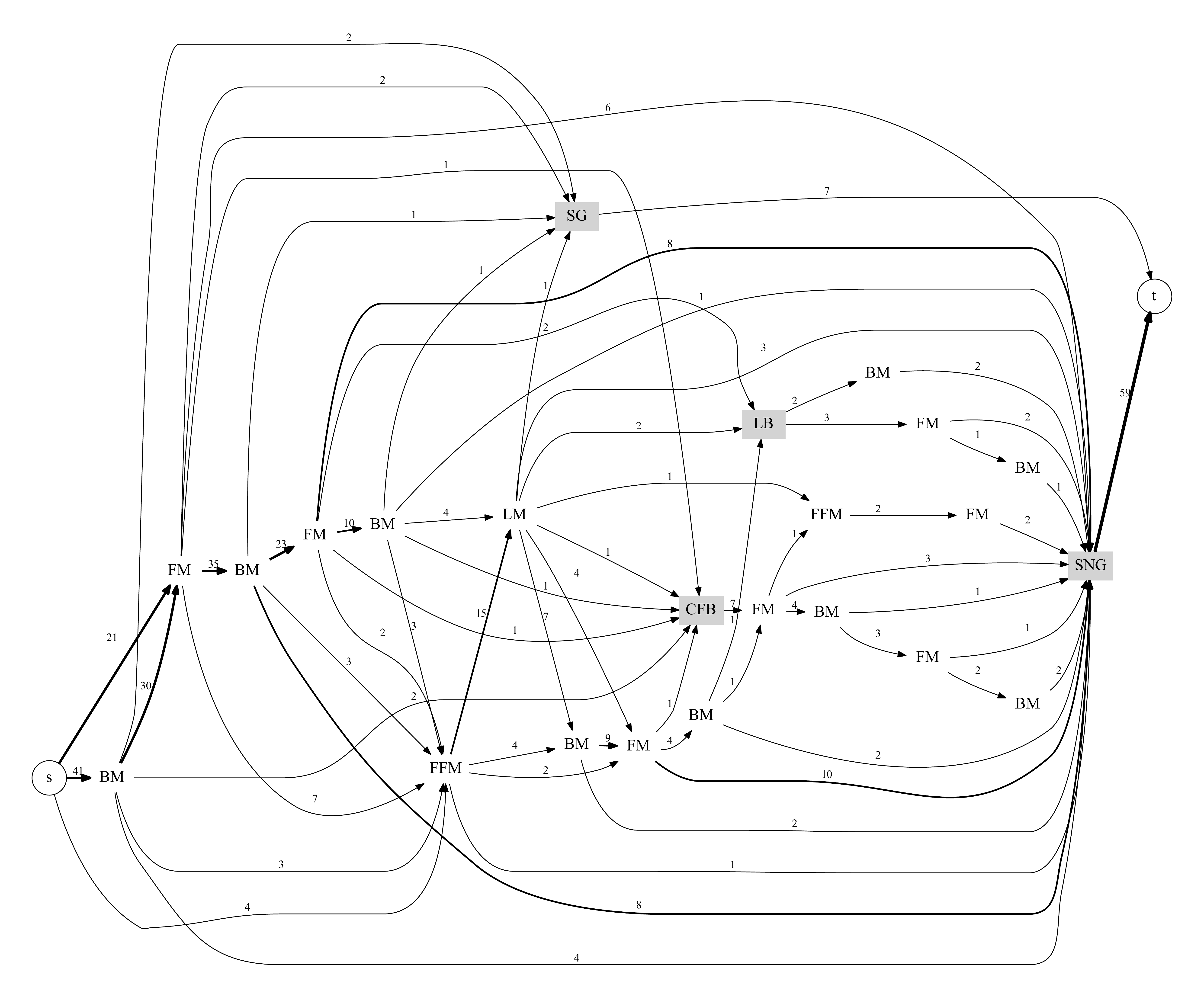}
    \caption{\Fd{} produced for the home team. The edge weights are the number of possessions that span the edge, and the nodes with grey background are event types that are significant, as defined in Section~\ref{sub:attacking-plays}.}
    \label{fig:flow-diagram-home-team}
\end{figure}
\begin{figure}[tb] 
    \centering
    \includegraphics[scale=0.23]{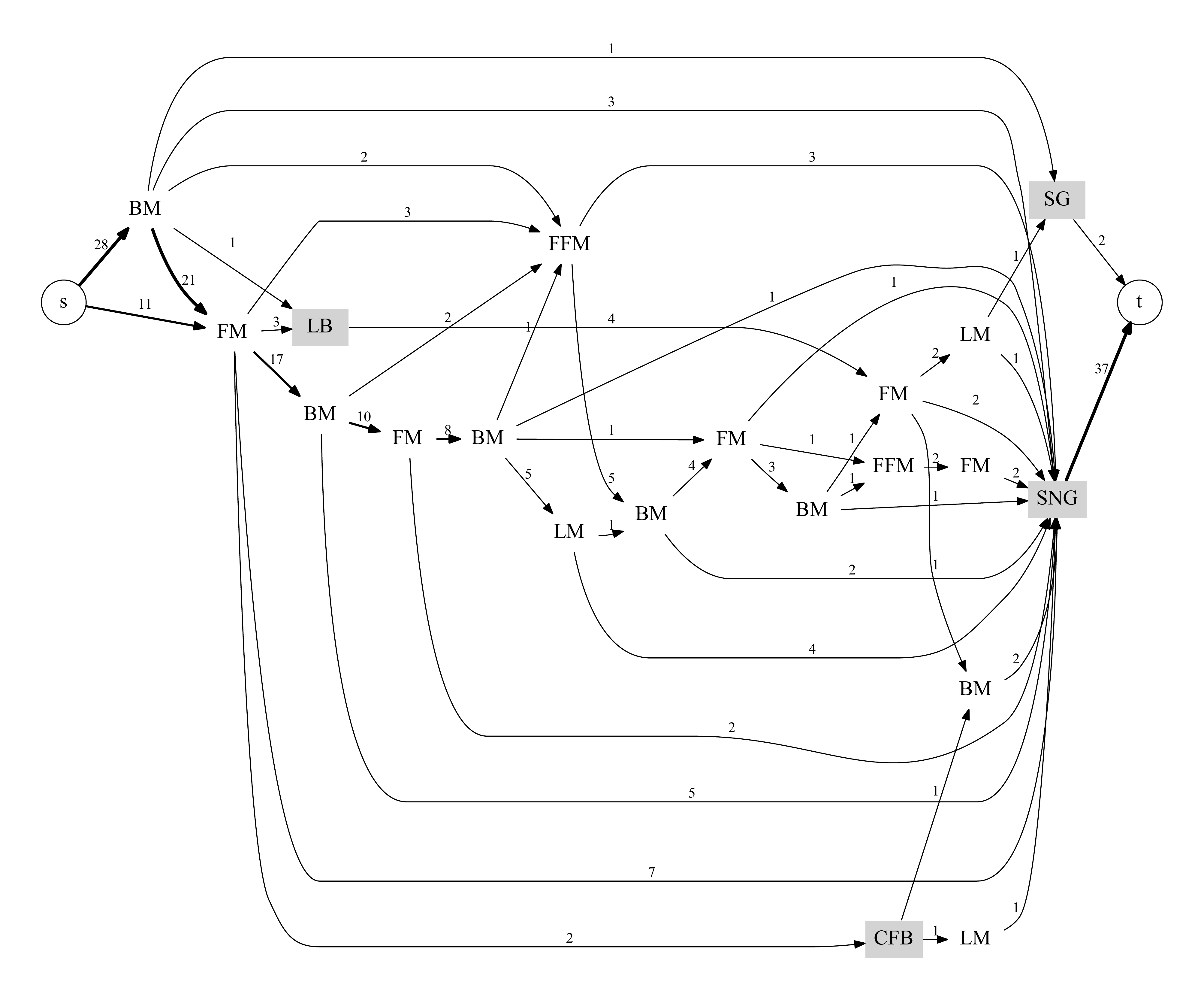}
    \caption{\Fd{} produced for the visiting team. The edge weights are the number of possessions that span the edge, and the nodes with grey background are event types that are significant, as defined in Section~\ref{sub:attacking-plays}.}
    \label{fig:flow-diagram-visiting-team}
\end{figure}

For a football analyst, the first four criteria are simple movements, and are not particularly interesting. 
The last four events are significant: the long ball and cross-field ball change the locus of attack; and the shot criteria represent the objective of an attack.

The possession state sequences for the home and visiting teams were segmented according to the criteria and the time-step heuristic algorithm was used to compute the \fd{}s.
The home-team input consisted of $66$ sequences covered by a total of $866$ segments, and resulted in a \fd{} with $25$ nodes and $65$ edges, see Fig.~\ref{fig:flow-diagram-home-team}. Similarly, the visiting-team input consisted of $39$ state sequences covered by $358$ segments and the output \fd{} complexity was $22$ nodes and $47$ edges, as shown in Fig.~\ref{fig:flow-diagram-visiting-team}.
        
%
%

At first glance, the differences between these \fd{}s may be difficult to appreciate, however closer inspection reveals several interesting observations.
The $s$--$t$ paths in the home-team \fd{} tend to be longer than those in the visiting team's, suggesting that the home team tends to retain possession of the ball for longer, and varies the intensity of attack more often.
Moreover, the nodes for cross-field passes and long-ball passes tend to occur earlier in the $s$--$t$ paths in the visiting team's \fd{}.
These are both useful tactics as they alter the locus of attack, however they also carry a higher risk.
This suggests that the home team is more confident in its ability to maintain possession for long attack possessions, and will only resort to such risky tactics later in a possession.
Furthermore, the tactics used by the team in possession are also impacted by the defensive tactics.
As Bialkowski et al~\cite{bialkowski-2014} found, visiting teams tend to set their defence deeper, i.e. closer to the goal they are defending.
When the visiting team is in possession, there is thus likely to be more space behind the home team's defensive line, and the long ball may appear to be a more appealing tactic.
The observations made from these are consistent with our basic understanding of football tactics, and suggest that the \fd{}s are interpretable in this application domain.


\subsection{Performance Testing}

In the second experiment, we used a generator that outputs synthetic state sequences and segmentations, and tested the performance of the algorithms on inputs of varying sizes.

The segmentations were generated using Markov Chain Monte Carlo sampling.
Nodes representing the criteria set of size $k$ were arranged in a ring and a Markov chain constructed, such that each node had a transition probability of $0.7$ to remain at the node, $0.1$ to move to the adjacent node, and $0.05$ to move to the node two places away.
Segmentations were computed by sampling the Markov chain starting at a random node. 
Thus, simulated datasets of arbitrary size $m$, state sequence length $n$, criteria set size $k$ were generated.

We performed two tests on the generated segmentations. 
In the first, experiments were run on the four algorithms described in Section~\ref{sec:Algorithms} with varying configurations of $m$, $n$ and $k$ to investigate the impact of input size on the algorithm's performance.
The evaluation metric used was the CPU time required to generate the \fd{} for the input.
In the second test, we compared the total complexity of the output \fd{} produced by the two heuristic algorithms with the baseline complexity of the \fd{} produced by the exact algorithm for monotone increasing and independent criteria.

We repeated each experiment five times with different input sequences for each trial, and the results presented are the mean values of the metrics over the trials.
Limits were set such that the process was terminated if the CPU time exceeded $1$ hour, or the memory required exceeded $8$GB.


%
The results of the first test showed empirically that the exact algorithms have time and storage complexity consistent with the theoretical worst-case bounds, Fig.~\ref{fig:execution-metrics}~\emph{(top)}. 
The heuristic algorithms were subsequently run against larger test data sets to examine the practical limits of the input sizes, and were able to process larger input -- for example, an input of $k=128$, $m=32$ and $n=1024$ was tractable -- although the cost is that the resulting \fd{}s were suboptimal, but correct, in terms of their total complexity.

For the second test, we investigated the complexity of the \fd{} induced by inputs of varying parameterisations when using the heuristic algorithms. 
The objective was to examine how close the complexity was to the optimal complexity produced using an exact algorithm.
The inputs exhibited monotone decreasing and independent criteria, and thus the corresponding algorithm was used to produce the baseline.
Fig.~\ref{fig:execution-metrics}~\emph{(bottom)} summarises the results for varying input parameterisations. 
The complexity of the \fd{}s produced by the two heuristic algorithms are broadly similar, and increase at worst linearly as the input size increases.
Moreover, while the complexity is not optimal it appears to remain within a constant factor of the optimal, suggesting that the heuristic algorithms could produce usable \fd{}s for inputs where the exact algorithms are not tractable.

\begin{figure}[t!]
    \centering
    \begin{subfigure}[t]{0.32\textwidth}
        \includegraphics[width=1.0\textwidth]{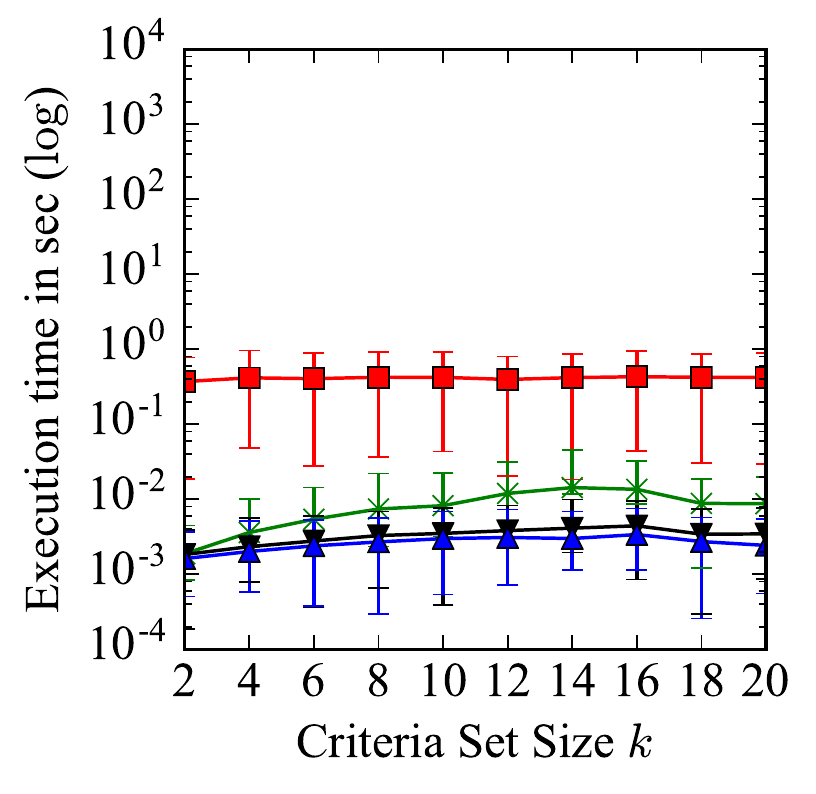}
    \end{subfigure}
    \begin{subfigure}[t]{0.32\textwidth}
        \includegraphics[width=1.0\textwidth]{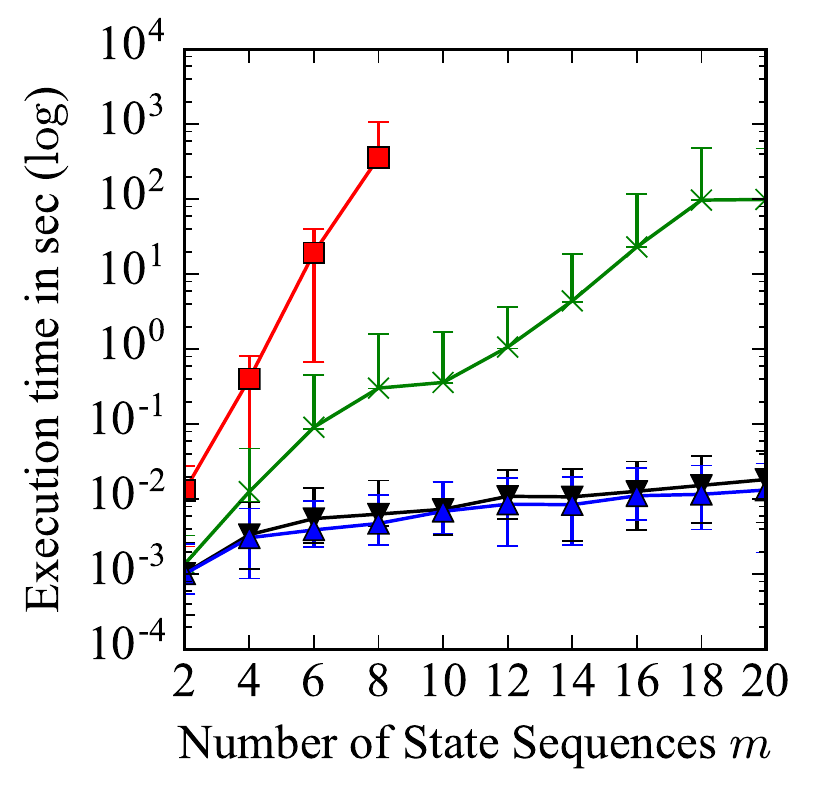}
    \end{subfigure}
    \begin{subfigure}[t]{0.32\textwidth}
        \includegraphics[width=1.0\textwidth]{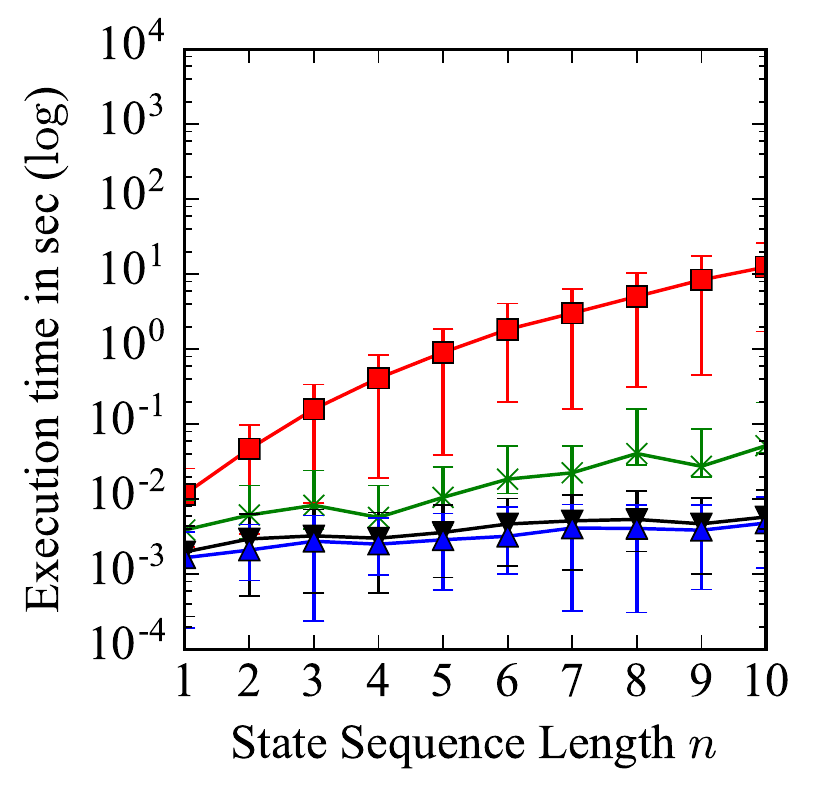}
    \end{subfigure}

    \begin{subfigure}[b]{1.0\textwidth}
        \centering
        \includegraphics[width=0.7\textwidth]{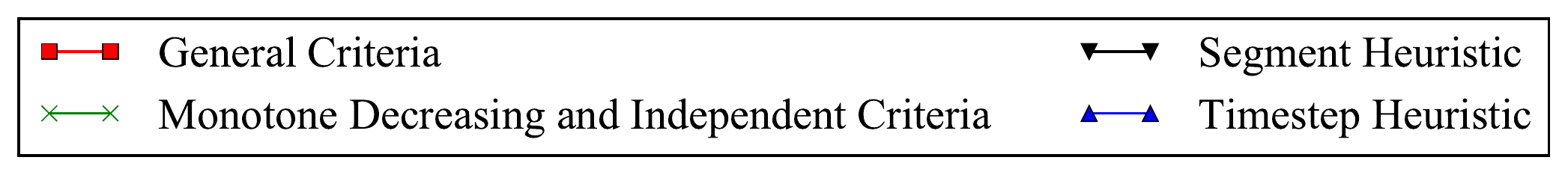}
    \end{subfigure}

    \begin{subfigure}[t]{0.32\textwidth}
        \includegraphics[width=1.0\textwidth]{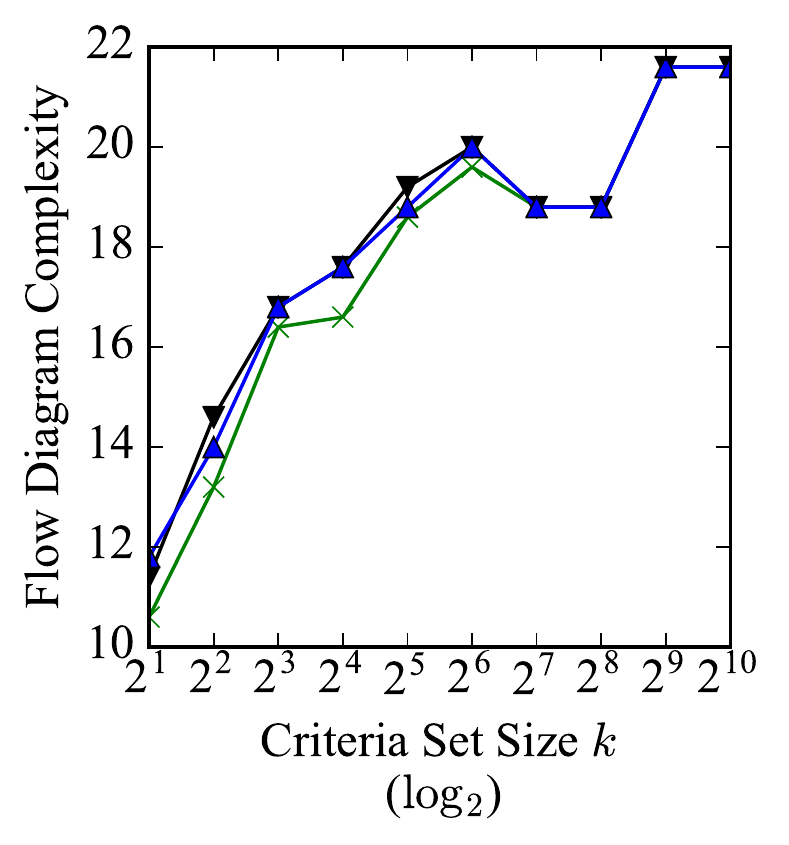}
    \end{subfigure}
    \begin{subfigure}[t]{0.32\textwidth}
        \includegraphics[width=1.0\textwidth]{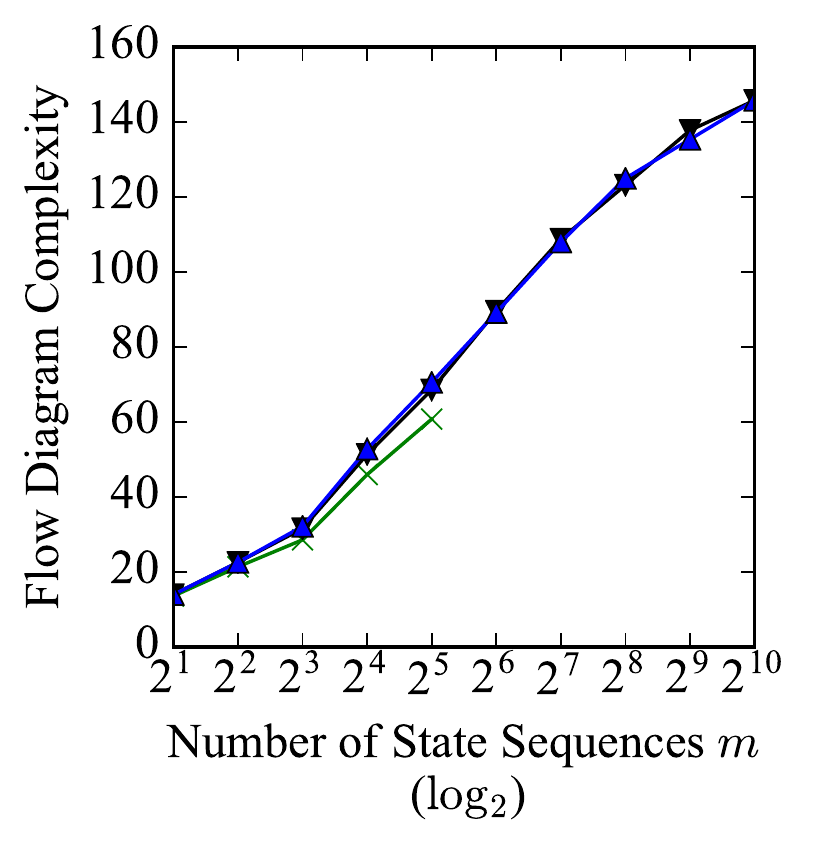}
    \end{subfigure}
    \begin{subfigure}[t]{0.32\textwidth}
        \includegraphics[width=1.0\textwidth]{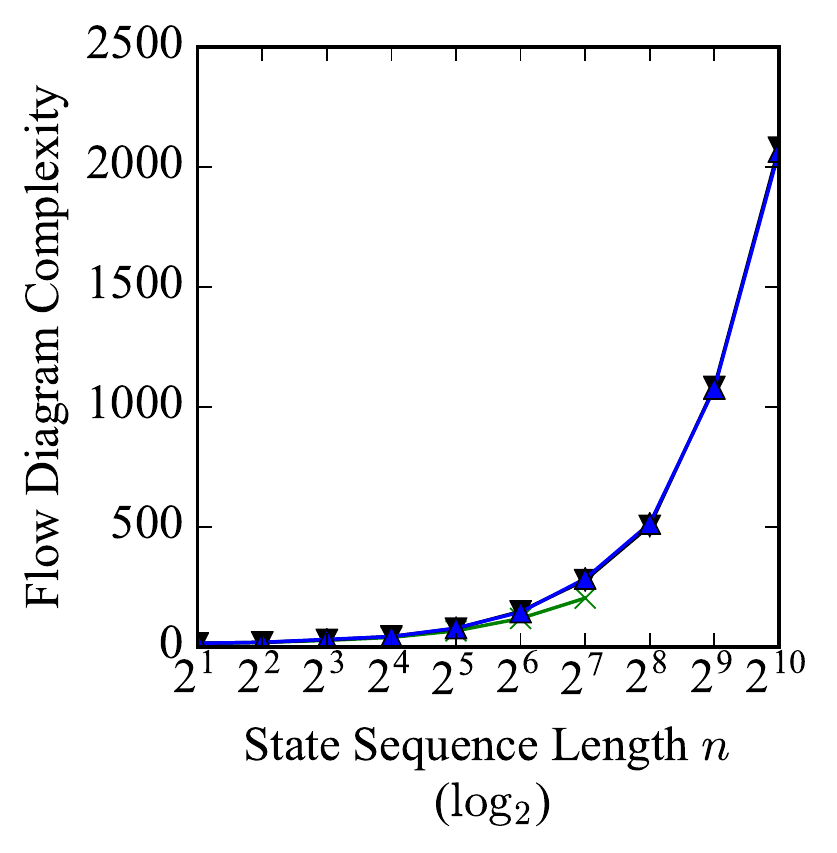}
    \end{subfigure}

    \caption{Runtime statistics for generating \protect{\fd{}} \emph{(top)}, and total complexity of flow diagrams produced \emph{(bottom)}. Default values of $m = 4$, $n = 4$ and $k = 10$ were used. The data points are the mean value and the error bars delimit the range of values over the five trials run for each input size.}
    \label{fig:execution-metrics}
\end{figure}

\section{Concluding Remarks}
We introduced \fd{}s as a compact representation of a large number of state
sequences.
We argued that this representation gives an intuitive summary allowing the user
to detect patterns among large sets of state sequences, and gave several
algorithms depending on the properties of the segmentation criteria. 
These algorithms only run in polynomial time if the number of state sequences $m$ is constant, which is the best we can hope for given the problem is $W[1]$-hard.
As a result we considered two heuristics capable of processing large data sets in reasonable time, however we were unable to give an approximation bound.
We tested the algorithms experimentally to assess the utility of the \fd{} representation in a sports analysis context, and also analysed the performance of the algorithms of inputs of varying parameterisations.

 \bibliographystyle{plain}
 \bibliography{flow-diagrams}

\newcommand{\SortNoop}[1]{}
\begin{thebibliography}{10}

\bibitem{abbkkw-tssc-14}
S.~P.~A. Alewijnse, K.~Buchin, M.~Buchin, A.~K{\"o}lzsch, H.~Kruckenberg, and
  M.~Westenberg.
\newblock A framework for trajectory segmentation by stable criteria.
\newblock In {\em Proc. 22nd ACM SIGSPATIAL/GIS}, pages 351--360. ACM, 2014.

\bibitem{adkls-stnmc-13}
B.~Aronov, A.~Driemel, M.~J. van Kreveld, M.~L{\"o}ffler, and F.~Staals.
\newblock Segmentation of trajectories for non-monotone criteria.
\newblock In {\em Proc. 24th ACM-SIAM SODA}, pages 1897--1911, 2013.

\bibitem{bialkowski-2014a}
A.~Bialkowski, P.~Lucey, G.~P.~K. Carr, Y.~Yue, S.~Sridharan, and I.~Matthews.
\newblock Identifying team style in soccer using formations learned from
  spatiotemporal tracking data.
\newblock In {\em {ICDM} Workshops}, pages 9--14. {IEEE}, 2014.

\bibitem{bialkowski-2014}
A.~Bialkowski, P.~Lucey, P.~Carr, Y.~Yue, and I.~Matthews.
\newblock Win at home and draw away: automatic formation analysis highlighting
  the differences in home and away team behaviors.
\newblock In {\em Proc. 8th Annual MIT Sloan Sports Analytics Conference},
  2014.

\bibitem{bbgll-dcpcs-11}
K.~Buchin, M.~Buchin, J.~Gudmundsson, M.~L{\"{o}}ffler, and J.~Luo.
\newblock Detecting commuting patterns by clustering subtrajectories.
\newblock {\em Int. J. Comput. Geometry Appl.}, 21(3):253--282, 2011.

\bibitem{bbkss-tgs-13}
K.~Buchin, M.~Buchin, M.~J. van Kreveld, B.~Speckmann, and F.~Staals.
\newblock Trajectory grouping structure.
\newblock In {\em Proc. 13th WADS}, pages 219--230, 2013.

\bibitem{bdks-stfa-11}
M.~Buchin, A.~Driemel, M.~van Kreveld, and V.~Sacristan.
\newblock Segmenting trajectories: A framework and algorithms using
  spatiotemporal criteria.
\newblock {\em Journal of Spatial Information Science}, 3:33--63, 2011.

\bibitem{bkk-stmb-12}
M.~Buchin, H.~Kruckenberg, and A.~K{\"o}lzsch.
\newblock Segmenting trajectories based on movement states.
\newblock In {\em Proc. 15th SDH}, pages 15--25. Springer-Verlag, 2012.

\bibitem{cwt-stdrd-06}
H.~Cao, O.~Wolfson, and G.~Trajcevski.
\newblock Spatio-temporal data reduction with deterministic error bounds.
\newblock {\em The VLDB Journal}, 15(3):211--228, 2006.

\bibitem{fi-aascs-95}
C.~B. Fraser and R.~W. Irving.
\newblock Approximation algorithms for the shortest common supersequence.
\newblock {\em Nordic Journal of Computing}, 2(3):303--325, 1995.

\bibitem{gudmundsson-2014}
J.~Gudmundsson and T.~Wolle.
\newblock {Football analysis using spatio-temporal tools}.
\newblock {\em Computers, Environment and Urban Systems}, 47:16--27, 2014.

\bibitem{hjzs-scabs-11}
C.-S. Han, S.-X. Jia, L.~Zhang, and C.-C. Shu.
\newblock Sub-trajectory clustering algorithm based on speed restriction.
\newblock {\em Computer Engineering}, 37(7), 2011.

\bibitem{kim2011}
H.-C. Kim, O.~Kwon, and K.-J. Li.
\newblock Spatial and spatiotemporal analysis of soccer.
\newblock In {\em Proc. 19th ACM SIGSPATIAL/GIS}, pages 385--388. ACM, 2011.

\bibitem{lucey-2013}
P.~Lucey, A.~Bialkowski, G.~P.~K. Carr, S.~Morgan, I.~Matthews, and Y.~Sheikh.
\newblock {Representing and Discovering Adversarial Team Behaviors Using Player
  Roles}.
\newblock In {\em Proc. IEEE Conference on Computer Vision and Pattern
  Recognition (CVPR'13)}, pages 2706--2713, Portland, OR, jun 2013. IEEE.

\bibitem{ly-hamp-93}
C.~Lund and M.~Yannakakis.
\newblock On the hardness of approximating minimization problems.
\newblock In {\em Proc. 25th ACM STOC}, pages 286--293. ACM, 1993.

\bibitem{p-pcfas-03}
K.~Pietrzak.
\newblock On the parameterized complexity of the fixed alphabet shortest common
  supersequence and longest common subsequence problems.
\newblock {\em Journal of Computer and System Sciences}, 67(4):757 -- 771,
  2003.

\bibitem{prozone}
{Prozone Sports Ltd}.
\newblock {Prozone Sports - Our technology}.
\newblock \url{http://prozonesports.stats.com/about/technology/}, 2015.

\bibitem{ru-scspb-81}
K.-J. Räihä and E.~Ukkonen.
\newblock The shortest common supersequence problem over binary alphabet is
  {NP}-complete.
\newblock {\em Theoretical Computer Sci.}, 16(2):187 -- 198, 1981.

\bibitem{vanhaaren-2015}
J.~{Van Haaren}, V.~Dzyuba, S.~Hannosset, and J.~Davis.
\newblock {Automatically Discovering Offensive Patterns in Soccer Match Data}.
\newblock In {\em Advances in Intelligent Data Analysis {XIV} - 14th
  International Symposium, {IDA} 2015}, volume 9385 of {\em Lecture Notes in
  Computer Science}, pages 286--297, Saint Etienne, oct 2015. Springer.

\bibitem{wang-2015}
Q.~Wang, H.~Zhu, W.~Hu, Z.~Shen, and Y.~Yao.
\newblock {Discerning Tactical Patterns for Professional Soccer Teams}.
\newblock In {\em Proceedings of the 21th ACM SIGKDD International Conference
  on Knowledge Discovery and Data Mining - KDD '15}, pages 2197--2206, Sydney,
  aug 2015. ACM Press.

\bibitem{wei-2013}
X.~Wei, L.~Sha, P.~Lucey, S.~Morgan, and S.~Sridharan.
\newblock {Large-Scale Analysis of Formations in Soccer}.
\newblock In {\em 2013 International Conference on Digital Image Computing:
  Techniques and Applications (DICTA)}, pages 1--8, Hobart, nov 2013. IEEE.

\end{thebibliography}

\newpage
\appendix

\end{document}